\documentclass[11pt]{article}

\usepackage{amsfonts}
\usepackage{amsmath}
\usepackage{amssymb}
\usepackage{hyperref}
\usepackage{multicol}

\usepackage{amsfonts}
\usepackage{amsmath}
\usepackage{amssymb}
\usepackage{enumerate}

\newcommand{\PROOF}{\begin{proof}}

\newcommand{\QED}{\end{proof}}

\newcommand{\ceil}[1]{ \left\lceil #1 \right\rceil }

\newcommand{\myset}[2]{ \left\{ #1 \;\left|\; #2 \right. \right\} }
\newcommand{\mysetl}[2]{ \left\{\left. #1 \right| #2 \right\} }

\newcommand{\prefix}{\sqsubseteq}

\newcommand{\limn}{\lim\limits_{n\to\infty}}
\newcommand{\liminfn}{\liminf\limits_{n\to\infty}}
\newcommand{\limsupn}{\limsup\limits_{n\to\infty}}

\newcommand{\N}{\mathbb{N}}
\newcommand{\Z}{\mathbb{Z}}
\newcommand{\Q}{\mathbb{Q}}
\newcommand{\R}{\mathbb{R}}

\newcommand{\FS}{{\mathrm{FS}}}

\newcommand{\constr}{{\mathrm{constr}}}

\newcommand{\dimfs}{\mathrm{dim}_\mathrm{FS}}

\newcommand{\cdim}{\mathrm{cdim}}

\renewcommand{\dim}{{\mathrm{dim}}}

\newcommand{\Dim}{{\mathrm{Dim}}}
\newcommand{\cDim}{{\mathrm{cDim}}}
\newcommand{\Dimfs}{{\Dim_\FS}}

\newcommand{\str}{{\mathrm{str}}}
\newcommand{\regSS}{S^\infty}
\newcommand{\strSS}{S^\infty_{\mathrm{str}}}

\newcommand{\G}{{\mathcal{G}}}

\newcommand{\D}{{\mathcal{D}}}

\newcommand{\FREQ}{{\mathrm{FREQ}}}
\newcommand{\freq}{{\mathrm{freq}}}

\newcommand{\strings}{\{0,1\}^*}

\newcommand{\K}{{\mathrm{K}}}

\newcommand{\CH}{\mathcal{H}}

\usepackage{amsthm}

\newtheorem{theorem}{Theorem}[section]
\newtheorem{corollary}[theorem]{Corollary}
\newtheorem{lemma}[theorem]{Lemma}
\newtheorem{proposition}[theorem]{Proposition}

\newtheorem{observation}[theorem]{Observation}

\theoremstyle{definition}

\newtheorem*{definition}{Definition}

\newtheorem{example}[theorem]{Example}

\newtheorem*{example*}{Example}
\newtheorem*{examples*}{Examples}

\newtheorem*{notation}{Notation}

\theoremstyle{remark}

\numberwithin{equation}{section}
\numberwithin{figure}{section}

\renewcommand{\include}{\input}

\newcommand{\I}{\mathcal{I}}

\begin{document}

\title{ {\bf
A Divergence Formula for Randomness and Dimension}
}
\author{
Jack H. Lutz\footnote{This research was supported in part by National Science Foundation
 Grants 9988483, 0344187, 0652569, and 0728806 and by the Spanish
 Ministry of Education and Science (MEC) and the European Regional
 Development Fund (ERDF) under project TIN2005-08832-C03-02.}\\
Department of Computer Science\\
Iowa State University\\
Ames, IA 50011, USA\\
lutz@cs.iastate.edu
}

\date{}

\maketitle

\begin{abstract}

If $S$ is an infinite sequence over a finite alphabet $\Sigma$ and $\beta$
is a probability measure on $\Sigma$, then the {\it dimension} of $ S$ with
respect to $\beta$, written $\dim^\beta(S)$, is a constructive version of
Billingsley dimension that coincides with the (constructive Hausdorff)
dimension $\dim(S)$ when $\beta$ is the uniform probability measure.  This
paper shows that $\dim^\beta(S)$ and its dual $\Dim^\beta(S)$, the {\it strong
dimension} of $S$ with respect to $\beta$, can be used in conjunction with
randomness to measure the similarity of two probability measures
$\alpha$ and $\beta$ on $\Sigma$.  Specifically, we prove that the
{\it divergence formula}
\[
 \dim^\beta(R) = \Dim^\beta(R) =\frac{\CH(\alpha)}{\CH(\alpha) + \D(\alpha || \beta)}
\]
holds whenever $\alpha$ and $\beta$ are computable, positive probability
measures on $\Sigma$ and $R \in \Sigma^\infty$ is random with respect to
$\alpha$.  In this formula, $\CH(\alpha)$ is the Shannon entropy of $\alpha$,
and $\D(\alpha||\beta)$ is the Kullback-Leibler divergence between $\alpha$
and $\beta$.  We also show that the above formula holds for all
sequences $R$ that are $\alpha$-normal (in the sense of Borel) when
$\dim^\beta(R)$ and $\Dim^\beta(R)$ are replaced by the more effective
finite-state dimensions $\dimfs^\beta(R)$ and $\Dimfs^\beta(R)$.  In the
course of proving this, we also prove finite-state compression
characterizations of $\dimfs^\beta(S)$ and $\Dimfs^\beta(S)$.

\end{abstract}

\section{Introduction}\label{se:1}
The constructive dimension $\dim(S)$ and the constructive strong
dimension $\Dim(S)$ of an infinite sequence $S$ over a finite alphabet
$\Sigma$ are constructive versions of the two most important classical
fractal dimensions, namely, Hausdorff dimension \cite{Haus19} and packing
dimension \cite{Tricot82,Sull84}, respectively.  These two constructive
dimensions, which were introduced in \cite{Lutz:DISS,Athreya:ESDAICC}, have been
shown to have the useful characterizations
\begin{equation}\label{eq:1_1}
       \dim(S) = \liminf_{w\rightarrow S}\frac{\K(w)}{|w|\log |\Sigma|}
\end{equation}
and
\begin{equation}\label{eq:1_2}
       \Dim(S) = \limsup_{w\rightarrow S}\frac{\K(w)}{|w|\log |\Sigma|},
\end{equation}
where the logarithm is base-$2$ \cite{Mayordomo:KCCCHD,Athreya:ESDAICC}.  In these equations,
$\K(w)$ is the Kolmogorov complexity of the prefix $w$ of $S$, i.e., the
{\it length in bits of the shortest program} that prints the string
w.  (See section 2.6 or \cite{LiVi97} for details.)  The numerators in
these equations are thus the {\it algorithmic information content} of
w, while the denominators are the ``naive'' information content of $w$,
also in bits.  We thus understand \eqref{eq:1_1} and \eqref{eq:1_2} to say that $\dim(S)$
and $\Dim(S)$ are the lower and upper {\it information densities} of the
sequence $S$.  These constructive dimensions and their analogs at other
levels of effectivity have been investigated extensively in recent
years \cite{EFD-bib}.

The constructive dimensions $\dim(S)$ and $\Dim(S)$ have recently been
generalized to incorporate a probability measure $\nu$ on the sequence
space $\Sigma^\infty$ as a parameter \cite{Lutz:DPSSF}.  Specifically, for each
such $\nu$ and each sequence $S \in \Sigma^\infty$, we now have the
constructive dimension $\dim^\nu(S)$ and the constructive strong
dimension $\Dim^\nu(S)$ of $S$ with respect to $\nu$.  (The first of these
is a constructive version of Billingsley dimension \cite{Bill60}.)  When
$\nu$ is the uniform probability measure on $\Sigma^\infty$, we have
$\dim^\nu(S) = \dim(S)$ and $\Dim^\nu(S) = \Dim(S)$.  A more interesting
example occurs when $\nu$ is the product measure generated by a
nonuniform probability measure $\beta$ on the alphabet $\Sigma$.  In this
case, $\dim^\nu(S)$ and $\Dim^\nu(S)$, which we write as $\dim^\beta(S)$ and
$\Dim^\beta(S)$, are again the lower and upper information densities of
S, but these densities are now measured with respect to unequal letter
costs.  Specifically, it was shown in \cite{Lutz:DPSSF} that
\begin{equation}\label{eq:1_3}
       \dim^\beta(S) = \liminf_{w\rightarrow S}\frac{\K(w)}{\I_\beta(w)}
\end{equation}
and
\begin{equation}\label{eq:1_4}
       \Dim^\beta(S) = \limsup_{w\rightarrow S}\frac{\K(w)}{\I_\beta(w)},
\end{equation}
where
\[ \I_\beta(w) = \sum_{i=0}^{|w|-1} \log\frac{1}{\beta(w[i])}\]
is the Shannon self-information of $w$ with respect to $\beta$.
These unequal letter costs $\log(1/\beta(a))$ for $a \in \Sigma$
can in fact be useful.  For example, the complete analysis of the
dimensions of individual points in self-similar fractals given by
\cite{Lutz:DPSSF} requires these constructive dimensions with a particular
choice of the probability measure $\beta$ on $\Sigma$.

In this paper we show how to use the constructive dimensions
$\dim^\beta(S)$ and $\Dim^\beta(S)$ in conjunction with randomness to
measure the degree to which two probability measures on $\Sigma$ are
similar.  To see why this might be possible, we note that the
inequalities
\[         0 \leq \dim^\beta(S) \leq \Dim^\beta(S) \leq 1\]
hold for all $\beta$ and $S$ and that the maximum values
\begin{equation}\label{eq:1_5}
             \dim^\beta(R) = \Dim^\beta(R) = 1
\end{equation}
are achieved whenever the sequence $R$ is random with respect to $\beta$.
It is thus reasonable to hope that, if $R$ is random with respect to
some other probability measure $\alpha$ on $\Sigma$, then $\dim^\beta(R)$ and
$\Dim^\beta(R)$ will take on values whose closeness to $1$ reflects the
degree to which $\alpha$ is similar to $\beta$.

This is indeed the case.  Our first main theorem says that the
{\it divergence formula}
\begin{equation}\label{eq:1_6}
 \dim^\beta(R) = \Dim^\beta(R) =\frac{\CH(\alpha)}{\CH(\alpha) + \D(\alpha||\beta)}
\end{equation}
holds whenever $\alpha$ and $\beta$ are computable, positive probability
measures on $\Sigma$ and $R \in \Sigma^\infty$ is random with respect to
$\alpha$.  In this formula, $\CH(\alpha)$ is the Shannon entropy of $\alpha$,
and $\D(\alpha||\beta)$ is the Kullback-Leibler divergence between $\alpha$
and $\beta$.  When $\alpha = \beta$, the Kullback-Leibler divergence
$\D(\alpha||\beta)$ is $0$, so \eqref{eq:1_6} coincides with \eqref{eq:1_5}.  When $\alpha$ and
$\beta$ are dissimilar, the Kullback-Leibler divergence $\D(\alpha||\beta)$
is large, so the right-hand side of \eqref{eq:1_6} is small.  Hence the
divergence formula tells us that, when $R$ is $\alpha$-random,
$\dim^\beta(R) = \Dim^\beta(R)$ is a quantity in $[0,1]$ whose closeness to
$1$ is an indicator of the similarity between $\alpha$ and $\beta$.

The proof of \eqref{eq:1_6}  serves as an outline of our other, more challenging
task, which is to prove that the divergence formula  \eqref{eq:1_6} also
holds for the much more effective {\it finite-state} $\beta$-{\it dimension}
$\dimfs^\beta(R)$ and {\it finite-state strong} $\beta$-{\it dimension}
$\Dimfs^\beta(R)$. (These dimensions,  defined in section \ref{sse:2_5}, are generalizations
of finite-state dimension and finite-state strong dimension, which were
introduced in \cite{Dai:FSD,Athreya:ESDAICC}, respectively.)

With this objective in mind, our second main theorem characterizes the
finite-state $\beta$-dimensions in terms of finite-state data compression.
Specifically, this theorem says that, in analogy with  \eqref{eq:1_3} and  \eqref{eq:1_4},
the identities
\begin{equation}\label{eq:1_7}
\dimfs^\beta(S) = \inf_C \liminf_{w\rightarrow S}\frac{|C(w)|}{\I_\beta(w)}
\end{equation}
and
\begin{equation}\label{eq:1_8}
       \dimfs^\beta(S) = \inf_C \limsup_{w\rightarrow S}\frac{|C(w)|}{\I_\beta(w)}
\end{equation}
hold for all infinite sequences $S$ over $\Sigma$.  The infima here are taken
over all information-lossless finite-state compressors (a model introduced
by Shannon \cite{Shannon48} and investigated extensively ever since) $C$ with output
alphabet ${0,1}$, and $|C(w)|$ denotes the number of bits that $C$ outputs when
processing the prefix $w$ of $S$.  The special cases of \eqref{eq:1_7} and \eqref{eq:1_8} in
which $\beta$ is the uniform probability measure on $\Sigma$, and hence
$\I_\beta(w) = |w| \log |\Sigma|$, were proven in \cite{Dai:FSD,Athreya:ESDAICC}.  In fact,
our proof uses these special cases as ``black boxes'' from which we derive
the more general \eqref{eq:1_7} and \eqref{eq:1_8}.

With \eqref{eq:1_7} and \eqref{eq:1_8} in hand, we prove our third main theorem.  This involves
the finite-state version of randomness, which was introduced by Borel \cite{Bore09}
long before finite-state automata were defined.  If $\alpha$ is a probability
measure on $\Sigma$, then a sequence $S \in \Sigma^\infty$ is $\alpha$-{\it normal}
in the sense of Borel if every finite string $w \in \Sigma^*$ appears with
asymptotic frequency $\alpha(w) \in S$, where we write
\[            \alpha(w) = \prod_{i=0}^{|w|-1} \alpha(w[i]).\]
(See section 2.6 for a precise definition of asymptotic frequency.)
Our third main theorem says that the {\it divergence formula}
\begin{equation}\label{eq:1_9}
\dimfs^\beta(R) = \Dimfs^\beta(R) =\frac{\CH(\alpha)}{\CH(\alpha) + \D(\alpha||\beta)}
\end{equation}
holds whenever $\alpha$ and $\beta$ are  positive probability
measures on $\Sigma$ and $R \in \Sigma^\infty$ is $\alpha$-normal.

In section \ref{se:2} we briefly review ideas from Shannon information
theory, classical fractal dimensions, algorithmic information theory,
and effective fractal dimensions that are used in this paper.  Section
\ref{se:3} outlines the proofs of \eqref{eq:1_6}, section \ref{se:4} outlines the
proofs of \eqref{eq:1_7} and \eqref{eq:1_8}, and section \ref{se:5} outlines the
proof of \eqref{eq:1_9}. Various proofs are consigned to a
technical appendix. 
\section{Preliminaries}\label{se:2}

\subsection{Notation and setting}\label{sse:2_1}

Throughout this paper we work in
a finite alphabet $\Sigma=\{ 0,1,\dots, k-1\}$, where
$k\geq 2$. We write $\Sigma^*$ for the set of
(finite) {\em strings} over $\Sigma$ and $\Sigma^\infty$ for
the set of (infinite) {\em sequences} over $\Sigma$.
We write $|w|$ for the length of a
string $w$ and $\lambda$ for the empty string.
For $w\in\Sigma^*$ and $0\leq i <|w|$, $w[i]$ is the
$i$th symbol in $w$. Similarly, for $S\in\Sigma^\infty$ and
$i\in\N$ ($=\{ 0,1,2,\dots\}$), $S[i]$ is the $i$th
symbol in $S$. Note that the leftmost symbol in a
string or sequence is the $0$th symbol.

A {\em prefix} of a string or sequence $x\in\Sigma^*\cup\Sigma^\infty$
is a string $w\in\Sigma^*$ for which there exists a
string or sequence $y\in\Sigma^*\cup\Sigma^\infty$ such that $x=wy$.
In this case we write $w\prefix x$. The equation
$\lim_{w\rightarrow S}f(w) =L$ means that, for all $\epsilon>0$, for
all sufficiently long prefixes $w\prefix S$, $|f(w)-L|<\epsilon$.
We also use the limit inferior,
\[\liminf_{w\rightarrow S} f(w)=\lim_{w\rightarrow S} \inf\myset{f(x)}{w\prefix x\prefix S},\]
and the limit superior
\[\limsup_{w\rightarrow S}f(w) =\lim_{w\rightarrow S}\sup\myset{f(x)}{w\prefix x\prefix S}.\]

\subsection{Probability measures, gales, and Shannon information}\label{sse:2_2}

A {\em probability measure} on $\Sigma$ is a function
$\alpha:\Sigma\rightarrow [0,1]$ such that $\sum_{a\in \Sigma }\alpha(a)=1$. A
probability measure $\alpha$ on $\Sigma$ is {\em positive} if
$\alpha(a)>0$ for every $\alpha\in\Sigma$. A probability
measure $\alpha$ on $\Sigma$ is {\em rational} if $\alpha(a) \in\Q$
(i.e., $\alpha(a)$ is a rational number) for every $a\in\Sigma$.

A {\em probability measure} on $\Sigma^\infty$ is a function
$\nu:\Sigma^*\rightarrow[0,1]$ such that $\nu(\lambda)=1$ and,
for all $w\in\Sigma^*$, $\nu(w) =\sum_{a\in\Sigma} \nu(wa)$. (Intuitively,
$\nu(w)$ is the probability that $w\prefix S$ when
the sequence $S\in\Sigma^\infty$ is ``chosen according to $\nu$.'')
Each probability measure $\alpha$ on $\Sigma$ naturally
induces the probability measure $\alpha$ on $\Sigma^\infty$
defined by
\begin{equation}\label{eq:2_1}
\alpha(w)=\prod_{i=0}^{|w|-1} \alpha(w[i])
\end{equation}
for all $w\in\Sigma^*$.

We reserve the symbol $\mu$ for the
{\em uniform probability measure} on $\Sigma$, i.e.,
\[\mu(a)=\frac{1}{k}\text{ for all $a\in\Sigma$},\]
and also for the {\em uniform probability measure}
on $\Sigma^\infty$, i.e.,
\[\mu(w) =k^{-|w|}\text{ for all $w\in\Sigma^*$}.\]

If $\alpha$ is a probability measure on $\Sigma$
and $s\in [0,\infty)$, then an $s$-$\alpha$-{\em gale} is a
function $d:\Sigma^*\rightarrow [0,\infty)$ satisfying
\begin{equation}\label{eq:2_2}
d(w)=\sum_{a\in\Sigma} d(wa)\alpha(a)^s
\end{equation}
for all $w\in\Sigma^*$. A $1$-$\alpha$-gale is also
called an $\alpha$-{\em martingale}. When $\alpha=\mu$, we
omit it from this terminology, so an
$s$-$\mu$-gale is called an $s$-{\em gale}, and a
$\mu$-martingale is called a {\em martingale}.

We frequently use the following simple fact
without explicit citation.

\begin{observation}\label{ob:2_1}
Let $\alpha$ and $\beta$ be positive
probability measures on $\Sigma$, and let $s,t\in[0,\infty)$.
If $d:\Sigma^*\rightarrow[0,\infty)$ is an $s$-$\alpha$-gale, then
the function $\tilde{d}:\Sigma^*\rightarrow [0,\infty)$ defined by
\[\tilde d(w)=\frac{\alpha(w)^s}{\beta(w)^t} d(w)\]
is a $t$-$\beta$-gale.
\end{observation}

Intuitively, an $s$-$\alpha$-gale is a strategy for
betting on the successive symbols in a
sequence $S\in\Sigma^\infty$. For each prefix $w\prefix S$,
$d(w)$ denotes the amount of capital (money)
that the gale $d$ has after betting on the
symbols in $w$. If $s=1$, then the right-hand side of \eqref{eq:2_2} is the conditional
expectation of $d(wa)$, given that $w$ has
occurred, so \eqref{eq:2_2} says that the
payoffs are fair. If $s<1$, then \eqref{eq:2_2}
says that the payoffs are unfair.

Let $d$ be a gale, and let $S\in\Sigma^\infty$. Then
$d$ {\em succeeds} on $S$ if $\limsup_{w\rightarrow S} d(w)=\infty$,
and $d$ {\em succeeds strongly} on $S$ if $\liminf_{w\rightarrow S} d(w)=\infty$.
The {\em success set} of $d$ is the set $\regSS[d]$
of all sequences on which $d$ succeeds,
and the {\em strong success set} of $d$ is the set
$\strSS[d]$ of all sequences on which $d$ succeeds
strongly.

The {\em Shannon entropy } of a probability
measure $\alpha$ on $\Sigma$ is
\[\CH(\alpha) = \sum_{a\in\Sigma}\alpha(a) \log \frac{1}{\alpha(a)},\]
where $0\log \frac{1}{0}=0$. (unless otherwise indicated,
all logarithms in this paper are base-$2$.)
The {\em Kullback-Leibler divergence} between two
probability measures $\alpha$ and $\beta$ on $\Sigma$ is
\[\D(\alpha||\beta) =\sum_{a\in\Sigma} \alpha(a)\log \frac{\alpha(a)}{\beta(a)}.\]
The Kullback-Leibler divergence is used to
quantify how ``far apart'' the two probability
measures $\alpha$ and $\beta$ are. The {\em Shannon
self-information} of a string $w\in\Sigma^*$ with
respect to a probability measure $\beta$ on $\Sigma$ is
\[\I_\beta(w) =\log \frac{1}{\beta(w)} =\sum_{i=0}^{|w|-1}\log \frac{1}{\beta(w[i])}.\]
Discussions of $\CH(\alpha)$, $\D(\alpha||\beta)$, $\I_\beta(w)$ and
their properties may be found in any good
text on information theory, e.g., \cite{CovTho06}.

\subsection{Hausdorff, packing, and Billingsley dimensions}\label{sse:2_3}

Given a probability measure $\beta$ on $\Sigma$, each
set $X\subseteq \Sigma^\infty$ has a {\em Hausdorff  dimension}
$\dim(X)$, a {\em packing dimension} $\Dim(X)$, a
{\em Billingsley dimension} $\dim^\beta(X)$, and a
{\em strong Billingsley dimension} $\Dim^\beta(X)$, all of
which are real numbers in the interval $[0,1]$.
In this paper we are not concerned with
the original definitions of these classical
dimensions, but rather in their recent
characterizations (which may be taken as
definitions) in terms of gales.

\begin{notation}
For each probability measure $\beta$ on $\Sigma$
and each set $X\subseteq \Sigma^\infty$, let $\G^\beta(X)$
(respectively, $\G^{\beta,\str}(X)$) be the set of all
$s\in[0,\infty)$ such that there is a $\beta$-$s$-gale
$d$ satisfying $X\subseteq \regSS[d]$ (respectively, $X\subseteq \strSS[d]$).

\begin{theorem}[gale characterizations of classical fractal dimensions]\label{th:2_2}
Let $\beta$ be a probability measure on $\Sigma$, and let $X\subseteq \Sigma^\infty$.
\begin{enumerate}
\begin{multicols}{2}
\item \cite{Lutz:DCC} $\dim(X)=\inf \G^\mu(X)$.
\item \cite{Athreya:ESDAICC} $\Dim(X)=\inf \G^{\mu,\str}(X)$.

\item \cite{Lutz:DPSSF} $\dim^\beta(X) =\inf \G^\beta(X)$.
\item \cite{Lutz:DPSSF} $\Dim^\beta(X)=\inf\G^{\beta,\str}(X)$.
\end{multicols}
\end{enumerate}
\end{theorem}
\end{notation}

\subsection{Randomness and constructive dimensions}\label{sse:2_4}

Randomness and constructive dimensions are
defined by imposing computability constraints
on gales.

A real-valued function $f:\Sigma^*\rightarrow \R$ is
{\em computable} if there is a computable,
rational-valued function $\hat f:\Sigma^*\times \N\rightarrow \Q$
such that, for all $w\in\Sigma^*$ and $r\in\N$,
\[|\hat f(w,r) -f(w)|\leq 2^{-r}.\]
A real-valued function $f:\Sigma^*\rightarrow\R$ is
{\em constructive}, or {\em lower semicomputable}, if
there is a computable, rational-valued function
$\hat f:\Sigma^*\times\N\rightarrow\Q$ such that
\begin{enumerate}[(i)]
\item for all $w\in\Sigma^*$ and $t\in\N$, $\hat f(w,t)\leq \hat f(w, t+1)< f(w)$, and
\item
for all $w\in\Sigma^*$, $f(w)=\lim_{t\rightarrow\infty} \hat f(w,t)$.
\end{enumerate}

The first successful definition of the
randomness of individual sequences $S\in\Sigma^\infty$
was formulated by Martin-L\"of \cite{Martinlof66}.
Many characterizations (equivalent definitions)
of randomness are now known, of which
the following is the most pertinent.
\begin{theorem}[Schnorr \cite{Schn71a,Schn77}]
Let $\alpha$ be a probability measure on $\Sigma$. A sequence
$S\in\Sigma^\infty$ is {\em random } with respect to $\alpha$ (or,
briefly, $\alpha$-{\em random}) if there is no
constructive $\alpha$-martingale that succeeds on $S$.
\end{theorem}

Motivated by Theorem \ref{th:2_2}, we now define the
constructive dimensions.

\begin{notation}
We define the sets $\G^\beta_\constr(X)$ and
$\G^{\beta,\str}_\constr(X)$ to be like the sets $\G^\beta(X)$ and
$\G^{\beta,\constr}(X)$ of section \ref{sse:2_3}, except that
the $\beta$-$s$-gales are now required to be
constructive.
\end{notation}

\begin{definition}
Let $\beta$ be a probability measure
on $\Sigma$, let $X\subseteq \Sigma^\infty$, and let $S\in\Sigma^\infty$.
\begin{enumerate}
\item \cite{Lutz:DISS} The {\em constructive dimension} of $X$
is $\cdim(X)=\inf \G^\mu_\constr(X)$.
\item \cite{Athreya:ESDAICC}
The {\em constructive strong dimension}
of $X$ is $\cDim(X)=\inf \G^{\mu,\str}_\constr(X)$.
\item \cite{Lutz:DPSSF}
The {\em constructive } $\beta$-{\em dimension}
of $X$ is $\cdim^\beta(X)=\inf \G^\beta_\constr(X)$.
\item \cite{Lutz:DPSSF} The {\em constructive strong} $\beta$-{\em dimension}
of $X$ is $\cDim^\beta(X)=\inf \G^{\beta,\str}_\constr(X)$.
\item
\cite{Lutz:DISS}
The {\em dimension} of $S$ is $\dim(S)=\cdim(\{S\})$.
\item
\cite{Athreya:ESDAICC}
The {\em strong dimension} of $S$ is
$\Dim(S)=\cDim(\{S\})$.
\item \cite{Lutz:DPSSF}
The $\beta$-{\em dimension} of $S$ is
$\dim^\beta(S)=\cdim^\beta(\{S\})$.
\item \cite{Lutz:DPSSF}
The {\em strong } $\beta$-{\em dimension} of $S$
is $\Dim^\beta(S)=\cDim^\beta(\{S\})$.
\end{enumerate}
\end{definition}

It is clear that definitions 1, 2, 5, and 6
above are the special case $\beta=\mu$ of definitions
3, 4, 7, and 8, respectively. It is known that
$\cdim^\beta(X)=\sup_{S\in X}\dim^\beta(S)$ and that $\cDim^\beta(X)
=\sup_{S\in X} \Dim^\beta(S)$ \cite{Lutz:DPSSF}.  Constructive dimensions are
thus investigated in terms of the dimensions
of individual sequences. Since one does not
discuss the classical dimension of an individual
sequence (because the dimensions of section \ref{sse:2_3}
are all zero for singleton, or even countable, sets),
no confusion results from the notation $\dim(S)$,
$\Dim(S)$, $\dim^\beta(S)$, and $\Dim^\beta(S)$.
\subsection{Normality and finite-state dimensions}\label{sse:2_5}

The preceding section developed the constructive
dimensions as effective versions of the classical
dimensions of section \ref{sse:2_3}. We now introduce the
even more effective finite-state dimensions.

\begin{notation}
$\Delta_\Q(\Sigma)$ is the set of all rational-valued probability measure on $\Sigma$.
\end{notation}

\begin{definition}[\cite{SchSti72,Fede91,Dai:FSD}]
A {\em finite-state gambler} ({\em FSG}) is a $4$-tuple
\[G=(Q, \delta, q_0, B),\]
where $Q$ is a finite set of {\em states}, $\delta:Q\times \Sigma\rightarrow Q$
is the {\em transition function}; $q_0\in Q$ is the {\em initial
state}, and $B:Q\rightarrow \Delta_\Q(\Sigma)$ is the {\em betting function}.
\end{definition}

The transition structure $(Q, \delta, q_0)$ here
works as in any deterministic finite-state automaton.
For $w\in\Sigma^*$, we write $\delta(w)$ for the state
reached by starting at $q_0$ and processing $w$
according to $\delta$.

Intuitively, if the above FSG is in state
$q\in Q$, then, for each $a\in \Sigma$, it bets the
fraction $B(q)(a)$ of its current capital
that the next input symbol is an $a$. The
payoffs are determined as follows.

\begin{definition}
Let $G=(Q, \delta, q_0, B)$ be an FSG.
\begin{enumerate}
\item The {\em martingale} of $G$ is the function
$d_G:\Sigma^*\rightarrow [0,\infty)$ defined by the recursion
\[d_G(\lambda)=1,\]
\[d_G(wa)=kd_G(w)B(\delta(w))(a)\]
for all $w\in\Sigma^*$ and $a\in \Sigma$.
\item
If $\beta$ is a probability measure on $\Sigma$
and $s\in[0,\infty)$, then the $s$-$\beta$-{\em gale} of $G$
is the function $d_{G,\beta}^{(s)}:\Sigma^*\rightarrow [0,\infty)$ defined
by
\[d_{G,\beta}^{(s)}(w)=\frac{\mu(w)}{\beta(w)^s}d_G(w)\]
for all $w\in \Sigma^*$.
\end{enumerate}
\end{definition}

It is easy to verify that $d_G = d_{G, \mu}^{(1)}$ is a
martingale. It follows by Observation \ref{ob:2_1}
that $d_{G,\beta}^{(1)}$ is an $s$-$\beta$-gale.

\begin{definition}
A {\em finite-state} $s$-$\beta$-{\em gale} is
an $s$-$\beta$-gale of the form $d_{G,\beta}^{(s)}$ for
some FSG $G$.
\end{definition}

\begin{notation}
We define the sets $\G_\mathrm{FS}^\beta(X)$ and
$\G_\mathrm{FS}^{\beta,\str}(X)$ to be like the sets $\G^\beta(X)$
and $\G^{\beta,\str}(X)$ of section \ref{sse:2_3}, except that
the $s$-$\beta$-gales are now required to
be finite-state.
\end{notation}

\begin{definition}
Let $\beta$ be a probability measure
on $\Sigma$, and let $S\in\Sigma^\infty$.
\begin{enumerate}
\item \cite{Dai:FSD} The {\em finite-state dimension} of $S$
is $\dimfs(S)=\inf \G^\mu_\mathrm{FS}(\{S\})$.
\item \cite{Athreya:ESDAICC} The {\em finite-state strong dimension}
of $S$ is $\Dimfs(S)=\inf \G_\mathrm{FS}^{\mu,\str}(\{S\})$.
\item The {\em finite-state $\beta$-dimension} of $S$
is $\dimfs^\beta(S)=\inf \G_\mathrm{FS}^\beta(\{S\})$.
\item
The {\em finite-state strong $\beta$-dimension} of
$S$ is $\Dimfs^\beta(S) =\inf \G_\mathrm{FS}^{\beta, \str}(\{S\})$.
\end{enumerate}
\end{definition}

We now turn to some ideas based on
asymptotic frequencies of strings in a given
sequence. For nonempty strings $w, x\in\Sigma^*$,
we write
\[\#_\square(w,x)=\left\lvert \mysetl{m\leq \frac{|x|}{|w|}-1}{x[m|w|..(m+1)|w|-1]=w} \right\rvert\]
for the number of block occurrences of $w$ in $x$.
For each sequence $S\in\Sigma^\infty$, each positive integer $n$,
and each nonempty $w\in \Sigma^{<n}$, the $n$th {\em block frequency} of
$w$ in $S$ is
\[\pi_{S,n}(w) =\frac{\#_\square(w, S[0..n|w|-1])}{n}.\]
Note that, for each $n$ and $l$, the restriction $\pi_{S,n}^{(l)}$
of $\pi_{S,n}$ to $\Sigma^l$ is a probability measure on $\Sigma^l$.

\begin{definition}
Let $\alpha$ be a probability measure on $\Sigma$,
let $S\in \Sigma^\infty$, and let $0<l\in \N$.
\begin{enumerate}
\item
$S$ is $\alpha$-$l$-{\em normal} in the sense of
Borel if, for all $w\in \Sigma^l$, $\limn \pi_{S,n}(w)=\alpha(w)$.
\item
$S$ is $\alpha$-{\em normal} in the sense of Borel
if $S$ is $\alpha$-$l$-normal for all $0<l\in\N$.
\item \cite{Bore09} $S$ is {\em normal } in the sense
of Borel if $S$ is $\mu$-normal.
\item
$S$ has {\em asymptotic frequency} $\alpha$, and we
write $S\in \FREQ^\alpha$, if $S$ is $\alpha$-$1$-normal.
\end{enumerate}
\end{definition}

\begin{theorem}[\cite{SchSti72,Bourke:ERFSD}]\label{th:2_4}
For each
probability measure $\alpha$ on $\Sigma$ and each $S\in\Sigma^\infty$,
the following three conditions are equivalent.
\begin{enumerate}[(1)]
\item
$S$ is $\alpha$-normal.
\item
No finite-state $\alpha$-martingale succeeds on $S$.
\item
$\dimfs^\alpha(S)=1$.
\end{enumerate}
\end{theorem}

The equivalence of (1) and (2) where $\alpha=\mu$ was
proven in \cite{SchSti72}. The equivalence of (2)
and (3) when $\alpha=\mu$ was noted in \cite{Bourke:ERFSD}.
The extensions of these facts to arbitrary $\alpha$ is
routine.

\newcommand{\mH}{\mathrm{H}}

For each $S\in\Sigma^\infty$ and $0<l\in\N$, the $l$th
{\em normalized lower and upper block entropy}
rates of $S$ are
\[\mH_l^-(S)=\frac{1}{l\log k}\liminfn \CH(\pi_{S,n}^{(l)})\]
and
\[\mH_l^+(S)=\frac{1}{l\log k}\limsupn \CH(\pi_{S,n}^{(l)}),\]
respectively.

We use the following result in section \ref{se:5}.
\begin{theorem}[\cite{Bourke:ERFSD}]\label{th:2_5}
Let $S\in\Sigma^\infty$.
\begin{enumerate}
\item $\dimfs(S)=\inf_{0<l\in\N}\mH_l^-(S)$.
\item
$\Dimfs(S)=\inf_{0<l\in\N}\mH_l^+(S)$.
\end{enumerate}
\end{theorem}

\subsection{Kolmogorov complexity and finite-state compression}\label{sse:2_6}
 We now review known characterizations of constructive
 and finite-state dimensions that are based on
 data compression ideas.

The {\em Kolmogorov complexity} $\K(w)$ of a string $w\in\Sigma^*$
is the minimum length of a program $\pi\in \strings$ for
which $U(\pi)=w$, where $U$ is a fixed universal
self-delimiting Turing machine \cite{LiVi97}.

\begin{theorem}\label{th:2_6}
Let $\beta$ be a probability measure
on $\Sigma$, and let $S\in\Sigma^\infty$.
\begin{enumerate}
\begin{multicols}{2}
\item \cite{Mayordomo:KCCCHD} $\dim(S) =\liminf_{w\rightarrow S} \frac{\K(w)}{|w|\log k}$.
\item \cite{Athreya:ESDAICC} $\Dim(S) =\limsup_{w\rightarrow S}\frac{\K(w)}{|w|\log k}$.

\item \cite{Lutz:DPSSF} $\dim^\beta(S)=\liminf_{w\rightarrow S} \frac{\K(w)}{\I_\beta(w)}$.
\item \cite{Lutz:DPSSF} $\Dim^\beta(S) =\limsup_{w\rightarrow S} \frac{\K(w)}{\I_\beta(w)}$.
\end{multicols}
\end{enumerate}
\end{theorem}

\begin{definition}[\cite{Shannon48}]
\begin{enumerate}
\item
A {\em finite-state compressor} (FSC) is a $4$-tuple
\[C=(Q,\delta, q_0, \nu),\]
where $Q$, $\delta$, and $q_0$ are as in the FSG definition,
and $\nu:Q\times\Sigma\rightarrow \strings$ is the output function.
\item
The {\em output} of an FSC $C=(Q,\delta, q_0,\nu)$
on an input $w\in\Sigma^*$ is the string $C(w)\in \strings$
defined by the recursion
\[C(\lambda) =\lambda,\]
\[C(wa)= C(w)\nu(\delta(w),a)\]
for all $w\in\Sigma^*$ and $a\in\Sigma$.
\item
An {\em information-lossless finite-state compressor} (ILFSC) is an FSC for which
the function
\[\Sigma^*\rightarrow\strings \times Q\]
\[w\mapsto (C(w),\delta(w))\]
is one-to-one.
\end{enumerate}
\end{definition}

\begin{theorem}\label{th:2_7}
Let $S\in\Sigma^\infty$.
\begin{enumerate}
\item \cite{Dai:FSD}
$\dimfs(S)=\inf_C \liminf_{w\rightarrow S} \frac{|C(w)|}{|w|\log k}$.
\item \cite{Athreya:ESDAICC} $\Dimfs(S)=\inf_C \limsup_{w\rightarrow S} \frac{|C(w)|}{|w|\log k}$.
\end{enumerate}
\end{theorem}

\section{Divergence formula for randomness and constructive dimensions}\label{se:3}

This section proves the divergence formula for
$\alpha$-randomness, constructive $\beta$-dimension, and
constructive strong $\beta$-dimension.
The key point
here is that  the Kolmogorov complexity characterizations
of these $\beta$-dimensions reviewed in section \ref{sse:2_6} immediately
imply the following fact.

\begin{lemma}\label{lm:3_1}
If $\alpha$ and $\beta$ are computable, positive
probability measure on $\Sigma$, then, for all $S\in\Sigma^\infty$,
\[\liminf_{w\rightarrow S} \frac{\I_\alpha(w)}{\I_\beta(w)}\leq \frac{\dim^\beta(S)}{\dim^\alpha(S)}\leq \limsup_{w\rightarrow S} \frac{\I_\alpha(w)}{\I_\beta(w)},\]
and
\[\liminf_{w\rightarrow S} \frac{\I_\alpha(w)}{\I_\beta(w)}\leq \frac{\Dim^\beta(S)}{\Dim^\alpha(S)}\leq \limsup_{w\rightarrow S} \frac{\I_\alpha(w)}{\I_\beta(w)}.\]
\end{lemma}

The following lemma is crucial to our
argument, both here and in section \ref{se:5}.

\begin{lemma}[frequency divergence lemma]\label{lm:3_2}
If $\alpha$ and $\beta$ are positive probability measures on
$\Sigma$, then, for all $S\in\FREQ^\alpha$,
\[\I_\beta(w)=(\CH(\alpha) +\D(\alpha ||\beta))|w| + o(|w|)\]
as $w\rightarrow S$.
\end{lemma}

The next lemma gives a simple relationship
between the constructive $\beta$-dimension and the
constructive dimension of any sequence that is
$\alpha$-$1$-normal.

\begin{lemma}\label{lm:3_3}
If $\alpha$ and $\beta$ are computable,
positive probability measures on $\Sigma$, then,
for all $S\in\FREQ^\alpha$,
\[\dim^\beta(S)=\frac{\dim(S)}{\CH_k(\alpha) +\D_k(\alpha||\beta)},\]
and
\[\Dim^\beta(S)=\frac{\Dim(S)}{\CH_k(\alpha) +\D_k(\alpha||\beta)}.\]
\end{lemma}

We now recall the following constructive
strengthening of a 1949 theorem of Eggleston \cite{Eggl49}.

\begin{theorem}[\cite{Lutz:DISS,Athreya:ESDAICC}]\label{th:3_4}
If $\alpha$ is a
computable probability measure on $\Sigma$, then, for
every $\alpha$-random sequence $R\in\Sigma^\infty$,
\[\dim(R)=\Dim(R) =\CH_k(\alpha).\]
\end{theorem}

The main result of this section is now clear.

\begin{theorem}[divergence formula for randomness and constructive dimensions]\label{th:3_5}
If $\alpha$ and $\beta$
are computable, positive probability measures
on $\Sigma$, then, for every $\alpha$-random sequence
$R\in\Sigma^\infty$,
\[\dim^\beta(R)=\Dim^\beta(R) =\frac{\CH(\alpha)}{\CH(\alpha)+\D(\alpha||\beta)}.\]
\end{theorem}
\begin{proof}
This follows immediately from Lemma \ref{lm:3_3}
and Theorem \ref{th:3_4}.
\end{proof}

We note that $\D(\alpha||\mu) =\log k -\CH(\alpha)$, so
Theorem \ref{th:3_4} is the case $\beta=\mu$ of Theorem \ref{th:3_5}. 
\section{Finite-state dimensions and data compression}\label{se:4}
This section proves finite-state compression
characterizations of finite-state $\beta$-dimension
and finite-state strong $\beta$-dimension that are
analogous to the characterizations given by parts 3 and 4 of
Theorem \ref{th:2_6}. Our argument uses the
following two technical lemmas, which are
proven in the technical appendix.

\begin{lemma}\label{lm:4_1}
Let $\beta$ be a positive probability measure on $\Sigma$, and let $C$ be an ILFSC.
Assume that $I\subseteq \Sigma^*$, $s>0$, and $\epsilon >0$
have the property that, for all $w\in I$,
\begin{equation}\label{eq:4_1}
s\geq \frac{|C(w)|}{\I_\beta(w)} +\epsilon.
\end{equation}
Then there exist an FSG $G$ and a real number
$\delta>0$ such that, for all sufficiently
long
strings $w\in I$,
\begin{equation}\label{eq:4_2}
d_{G,\beta}^{(s)}(w)\geq 2^{\delta|w|}.
\end{equation}
\end{lemma}

\begin{lemma}\label{lm:4_2}
Let $\beta$ be a positive probability
measure on $\Sigma$, and let $G$ be an FSG.
Assume that $I\subseteq \Sigma^*$, $s>0$, and $\epsilon >0$ have
the property that, for all $w\in I$,
\begin{equation}\label{eq:4_3}
d_{G,\beta}^{(s-2\epsilon)}(w)\geq 1.
\end{equation}
Then there is an ILFSC $C$ such that,
for all $w\in I$,
\begin{equation}\label{eq:4_4}
|C(w)|\leq s\I_\beta(w).
\end{equation}
\end{lemma}
We now prove the main result of
this section.

\begin{theorem}[compression characterizations of finite-state $\beta$-dimensions]\label{th:4_3}
If $\beta$ is a positive probability measure
on $\Sigma$, then, for each sequence $S\in\Sigma^\infty$,
\begin{equation}\label{eq:4_5}
\dimfs^\beta(S)=\inf_C\liminf_{w\rightarrow S} \frac{|C(w)|}{ \I_\beta(w)},
\end{equation}
and
\begin{equation}\label{eq:4_6}
\Dimfs^\beta(S)=\inf_C\limsup_{w\rightarrow S} \frac{|C(w)|}{\I_\beta(w)},
\end{equation}
where the infima are taken over all ILFCSs $C$.
\end{theorem}
\begin{proof}
Let $\beta$ and $S$ be as given. We
first prove that the left-hand sides of
\eqref{eq:4_5} and \eqref{eq:4_6} do not exceed the
right-hand sides. For this, let $C$
be an ILFSC. It suffices to show that
\begin{equation}\label{eq:4_7}
\dimfs^\beta(S)\leq \liminf_{w\rightarrow S} \frac{|C(w)|}{\I_\beta(w)}
\end{equation}
and
\begin{equation}\label{eq:4_8}
\Dimfs^\beta(S)\leq \limsup_{w\rightarrow S} \frac{|C(w)|}{\I_\beta(w)}.
\end{equation}

To see that \eqref{eq:4_7} holds, let $s$ exceed the right-hand side. Then
there exist an infinite set $I$ of
prefixes of $S$ and an $\epsilon >0$ such that
\eqref{eq:4_1} holds for all $w\in I$. It follows by
Lemma \ref{lm:4_1} that there exist an FSG $G$
and a $\delta>0$ such that, for all
sufficiently long $w\in I$, $d_{G,\beta}^{(s)}(w)\geq 2^{\delta |w|}$.
Since $I$ is infinite and $\delta >0$, this implies
that $S\in \regSS[d_{G,\beta}^{(s)}]$,  whence $\dimfs^\beta(S)\leq s$.
This establishes \eqref{eq:4_7}.

The proof that \eqref{eq:4_8} holds is
identical to the preceding paragraph,
except that $I$ is now a cofinite
set of prefixes of $S$, so $S\in \strSS[d_{G,\beta}^{(s)}]$.

It remains to be shown that the
right-hand sides of \eqref{eq:4_5} and \eqref{eq:4_6}
do not exceed the left-hand sides.
To see this for \eqref{eq:4_5}, let $s>\dimfs^\beta(S)$. It
suffices to show that there is an ILFSC $C$
such that
\begin{equation}\label{eq:4_9}
\liminf_{w\rightarrow S} \frac{|C(w)|}{\I_\beta(w)} \leq s.
\end{equation}
By our choice of $s$ there exists $\epsilon >0$
such that $s -2\epsilon > \dimfs^\beta(S)$. This implies
that there is an infinite set $I$ of
prefixes of $S$ such that \eqref{eq:4_3} holds for
all $w\in I$. Choose $C$ for $G$, $I$, $S$, and $\epsilon$ as
in Lemma \ref{lm:4_2}. Then
\begin{equation}\label{eq:4_10}
\liminf_{w\rightarrow S} \frac{|C(w)|}{\I_\beta(w)}\leq \inf_{w\in I} \frac{|C(w)|}{\I_\beta(w)}\leq s
\end{equation}
by \eqref{eq:4_4}, so \eqref{eq:4_9} holds.

The proof that the right-hand side of \eqref{eq:4_6} does not exceed the left-hand
side is identical to the preceding paragraph,
except that the limits inferior in \eqref{eq:4_9} and
\eqref{eq:4_10} are now limits superior, and the set
$I$ is now a cofinite set of prefixes of $S$.
\end{proof} 
\section{Divergence formula for normality and finite-state dimensions}\label{se:5}
This section proves the divergence formula for
$\alpha$-normality, finite-state $\beta$-dimension, and
finite-state strong $\beta$-dimension. As should now
be clear, Theorem \ref{th:4_3} enables us to
proceed in analogy with section \ref{se:3}.

\begin{lemma}\label{lm:5_1}
If $\alpha$ and $\beta$ are
positive probability measures on $\Sigma$, then,
for all $S\in\Sigma^\infty$,
\begin{equation}\label{eq:5_1}
\liminf_{w\rightarrow S}\frac{\I_\alpha(w)}{\I_\beta(w)}\leq \frac{\dimfs^\beta(S)}{\dimfs^\alpha(S)}\leq \limsup_{w\rightarrow S}\frac{\I_\alpha(w)}{\I_\beta(w)},
\end{equation}
and
\begin{equation}\label{eq:5_2}
\liminf_{w\rightarrow S}\frac{\I_\alpha(w)}{\I_\beta(w)}\leq \frac{\Dimfs^\beta(S)}{\Dimfs^\alpha(S)}\leq \limsup_{w\rightarrow S}\frac{\I_\alpha(w)}{\I_\beta(w)}.
\end{equation}
\end{lemma}
\begin{lemma}\label{lm:5_2}
If $\alpha$ and $\beta$ are
positive probability measures on $\Sigma$, then,
for all $S\in\FREQ^\alpha$,
\[\dimfs^\beta(S)=\frac{\dimfs(S)}{\CH_k(\alpha)+\D_k(\alpha|| \beta)},\]
and
\[\Dimfs^\beta(S) =\frac{\Dimfs(S)}{\CH_k(\alpha) + \D_k(\alpha||\beta)}.\]
\end{lemma}

We next prove a finite-state analog of
Theorem \ref{th:3_4}.

\begin{theorem}\label{th:5_3}
If $\alpha$ is a
probability measure on $\Sigma$, then, for every
$\alpha$-normal sequence $R\in\Sigma^\infty$,
\[\dimfs(R)=\Dimfs(R) =\CH_k(\alpha).\]
\end{theorem}

We now have our third main theorem.

\begin{theorem}[divergence theorem for normality and finite-state dimensions]\label{th:5_4}
If $\alpha$ and $\beta$ are   positive probability measures on $\Sigma$,
then, for every $\alpha$-normal sequence
$R\in\Sigma^\infty$,
\[\dimfs^\beta(R) =\Dimfs^\beta(R)=\frac{\CH(\alpha)}{\CH(\alpha) +\D(\alpha ||\beta)}.\]
\end{theorem}
\begin{proof}
This follows immediately from
Lemma \ref{lm:5_2} and Theorem \ref{th:5_3}.
\end{proof}

We again note that $\D(\alpha || \beta) =\log k -\CH(\alpha)$,
so Theorem \ref{th:5_3} is the case $\beta =\mu$ of
Theorem \ref{th:5_4}. 

\vspace{0.6cm}

\noindent{\bf Acknowledgments.}
I thank Xiaoyang Gu and Elvira Mayordomo for useful discussions.

\bibliographystyle{abbrv}
\bibliography{dim,rbm,main,random,dimrelated}

\newpage\appendix
\section{Appendix -- Various Proofs}

\begin{proof}[{\bf Proof of Lemma \ref{lm:3_2}}]
Assume the hypothesis, and let $S\in\FREQ^\alpha$.
Then, as $w\rightarrow S$, we have
\begin{align*}
\I_\beta(w)& = \sum_{i=0}^{|w|-1} \log\frac{1}{\beta(w[i])}\\
&= \sum_{a\in \Sigma} \#(a,w)\log\frac{1}{\beta(a)}\\
&=|w|\sum_{a\in\Sigma} \freq_a(w)\log\frac{1}{\beta(a)}\\
&=|w|\sum_{a\in\Sigma} (\alpha(a)+ o(1))\log\frac{1}{\beta(a)}\\
&=|w|\sum_{a\in\Sigma} \alpha(a)\log\frac{1}{\beta(a)}+ o(|w|)\\
&=|w|\sum_{a\in\Sigma} \left ( \alpha(a)\log\frac{1}{\alpha(a)} +\alpha(a)\log\frac{\alpha(a)}{\beta(a)} \right) +o(|w|)\\
&=( \CH(\alpha)+\D(\alpha ||\beta)) |w| + o(|w|).
\end{align*}
\end{proof}

\begin{proof}[{\bf Proof of Lemma \ref{lm:3_3}}]
Let $\alpha$, $\beta$, and $S$ be as given. By
the frequency divergence lemma, we have
\begin{align*}
\frac{\I_\mu(w)}{\I_\beta(w)}&=\frac{|w|\log k}{ (\CH(\alpha) +\D(\alpha||\beta))|w| +o(|w|)}\\
&=\frac{\log k }{\CH(\alpha)+\D(\alpha||\beta) +o(1)}\\
&=\frac{\log k }{\CH(\alpha)+\D(\alpha||\beta)} +o(1)\\
&=\frac{1}{\CH_k(\alpha) +\D_k(\alpha||\beta)} +o(1)
\end{align*}
as $w\rightarrow S$. The present lemma follows from
this and Lemma \ref{lm:3_1}.
\end{proof}

The following lemma summarizes the first
part of the proof of Theorem \ref{th:2_7}.

\begin{lemma}[\cite{Dai:FSD}]\label{lm:a_1}
For each ILFSC $C$ there
is an integer $m\in \Z^+$ such that, for each
$l\in \Z^+$, there is an FSG $G$ such that,
for all $w\in\Sigma^*$,
\begin{equation}\label{eq:a_1}
\log d_G^{(1)}(w) \geq |w| \log k -|C(w)| -m(\tfrac{|w|}{l} +l).
\end{equation}
\end{lemma}

\begin{proof}[{\bf Proof of Lemma \ref{lm:4_1}}]
Assume the hypothesis.
Let
\[\delta_\beta =\min_{a\in \Sigma}\log \frac{1}{ \beta(a)},\]
noting the following two things.
\begin{enumerate}[(i)]
\item
$\delta_\beta>0$, because $\beta $ is positive.
\item
For all $w\in \Sigma^*$,
\begin{equation}\label{eq:a_2}
\I_\beta(w)\geq \delta_\beta|w|.
\end{equation}
\end{enumerate}
Choose $m$ for $C$ as in Lemma \ref{lm:a_1},
let
\begin{equation}\label{eq:a_3}
l =\ceil{\frac{3m}{\epsilon \delta \beta}},
\end{equation}
and choose $G$ for $C$, $m$, and $l$ as in
Lemma \ref{lm:4_1}. Let
\[\delta =\tfrac{2}{3}\epsilon \delta_\beta,\]
noting that $\delta>0$ and that
\begin{align*}
|w|\geq l^2\implies& \epsilon\delta_\beta|w|-m(\tfrac{|w|}{l}+l)\\
&=\epsilon\delta_\beta|w| -\tfrac{m}{l}(|w| +l^2)\\
&\geq \epsilon \delta_\beta|w| -\tfrac{2m}{l} |w|\\
&= (\epsilon \delta_\beta -\tfrac{2m}{l}) |w|\\
&\geq ^{\eqref{eq:a_3}}\tfrac{2}{3}\epsilon \delta_\beta |w|,
\end{align*}
i.e., that
\begin{equation}\label{eq:a_4}
|w|\geq l^2\implies \epsilon \delta_\beta|w| -m(\tfrac{|w|}{l} +l )\geq \delta|w|.
\end{equation}
It follows that, for all $w\in I$ with $|w|\geq l^2$,
we have
\begin{align*}
\log d_{G,\beta}^{(s)}(w) &= \log (\tfrac{\mu(w)}{\beta(w)^s}d_G^{(1)}(w))\\
&= -|w|\log k +s \I_\beta(w) +\log d_G^{(1)}(w)\\
&\geq ^{\eqref{eq:a_1}} s \I_\beta(w) -|C(w)| -m(\tfrac{|w|}{l} +l)\\
&\geq ^{\eqref{eq:4_1}}s\I_\beta(w)  -m(\tfrac{|w|}{l}+l)\\
&\geq ^{\eqref{eq:a_2}} \epsilon\delta_\beta|w| -m(\tfrac{|w|}{l}+l)\\
&\geq ^{\eqref{eq:a_4}} \delta|w|.
\end{align*}
Hence \eqref{eq:4_2} holds.
\end{proof}

An FSG $G=(Q, \Sigma, \delta, \beta, q_0)$ is {\em nonvanishing}
if all its bets are nonzero, i.e., if $\beta(q)(a)>0$
holds for all $q\in Q$ and $a\in \Sigma$.

\begin{lemma}[\cite{Dai:FSD}]\label{lm:a_2}
For each FSG $G$ and
each $\delta >0$, there is a nonvanishing FSG $G'$
such that, for all $w\in \Sigma^*$,
\begin{equation}\label{eq:a_5}
d_{G'}^{(1)}(w) \geq k^{-\delta|w|}d_G^{(1)}(w).
\end{equation}
\end{lemma}

The following lemma summarizes the
second part of the proof of Theorem \ref{th:2_7}.

\begin{lemma}[\cite{Dai:FSD}]\label{lm:a_3}
For each nonvanishing FSG $G$ and
each $l\in \Z^+$, there exists an ILFSC $C$ such
that, for all $w\in \Sigma^*$,
\begin{equation}\label{eq:a_6}
|C(w)|\leq (1+\tfrac{2}{l})|w|\log k -\log d_G^{(1)} (w).
\end{equation}
\end{lemma}

\begin{proof}[{\bf Proof of Lemma \ref{lm:4_2}}]
Assume the hypothesis. Let
\[\gamma =\log \frac{1}{\beta_\mathrm{max}},\]
where
\[\beta_\mathrm{max} =\max_{a\in \Sigma} \beta(a).\]
Note that $\gamma>0$ (because $\beta $ is positive)
and that, for all $w\in\Sigma^*$,
\begin{equation}\label{eq:a_7}
\I_\beta(w)\geq \gamma|w|.
\end{equation}
Let
\begin{equation}\label{eq:a_8}
\delta=\frac{\gamma\epsilon }{\log k}
\end{equation}
and choose $G'$ for $G$ and $\delta$ as
in Lemma \ref{lm:a_2}. Let
\begin{equation}\label{eq:a_9}
l =\ceil{\frac{2\log k}{\gamma \epsilon}},
\end{equation}
and choose $C$ for $G'$ and $l$ as in
Lemma \ref{lm:a_3}. Then, for all $w\in I$,
\begin{align*}
|C(w)|
&\leq ^{\eqref{eq:a_6}} (1+\tfrac{2}{l}) |w|\log k -\log d_{G'}^{(1)}(w)\\
&\leq ^{\eqref{eq:a_9}} |w|(\gamma\epsilon +\log k ) -\log d_{G'}^{(1)}(w)\\
&\leq ^\eqref{eq:a_5} |w|(\gamma\epsilon +\log k) -\log(k^{-\delta|w|} d_G^{(1)}(w))\\
&=|w|(\gamma\epsilon+\log k +\delta\log k) -\log d_G^{(1)}(w)\\
&= |w|(2\gamma\epsilon +\log k) -\log d_G^{(1)}(w)\\
&=|w|(2\gamma\epsilon +\log k) - \log \left(\frac{\beta(w)^{s-2\epsilon}}{\mu(w)}d_{G,\beta}^{(s-2\epsilon)}(w)\right)\\
&\leq^\eqref{eq:a_3} |w|(2\gamma\epsilon+\log k) -\log \left(\frac{\beta(w)^{s-2\epsilon}}{\mu(w)}\right)\\
&=|w|(2\gamma\epsilon +\log k) -\log (k^{|w|}\beta(w)^{s-2\epsilon})\\
&= 2\gamma\epsilon|w| -\log \beta(w)^{s-2\epsilon}\\
&= 2\gamma\epsilon|w| +(s-2\epsilon) \I_\beta(w)\\
&\leq^\eqref{eq:a_7} s\I_\beta(w).
\end{align*}
\end{proof}

\begin{proof}[{\bf Proof of Lemma \ref{lm:5_2}}]
As in the proof of Lemma \ref{lm:3_3},
the hypothesis implies that
\[\frac{\I_\mu(w)}{\I_\beta(w)}=\frac{1}{\CH_k(\alpha) +\D_k(\alpha||\beta)}+ o(1)\]
as $w\rightarrow S$. The present lemma follows from
this and Lemma \ref{lm:5_1}.
\end{proof}

\begin{proof}[{\bf Proof of Theorem \ref{th:5_3}}]
Assume the hypothesis, and let $l\in\Z^+$.
Let $\alpha^{(l)}$ be the restriction of the product
probability measure $\mu^\alpha$ to $\Sigma^l$, noting that
$\CH(\alpha^{(l)}) = l \CH(\alpha)$. We first show that
\begin{equation}\label{eq:5_3}
\limn \CH(\pi^{(l)}_{R,n}) =\CH(\alpha^{(l)}),
\end{equation}
where $\pi^{(l)}_{R,n}$ is the empirical probability
measure defined in section \ref{sse:2_5}. To see this,
let $\epsilon >0$. By the continuity of the entropy
function, there is a real number $\delta>0$
such that, for all probability measures
$\pi$ on $\Sigma^l$,
\[\max_{w\in\Sigma^l} |\pi(w) -\alpha^{(l)}(w)|< \delta \implies |\CH(\pi) - \CH(\alpha^{(l)})|<\epsilon.\]
Since $R$ is  $\alpha$-normal, there is, for each $w\in\Sigma^l$,
a positive integer $n_w$ such that, for all $n\geq n_w$,
\[|\pi^{(l)}_{R,n}(w) -\alpha^{(l)}(w)| = |\pi^{(l)}_{R,n}(w) -\mu^\alpha (w) |< \delta.\]
Let $N =\max_{w\in\Sigma^l} n_w$. Then, for all $n\geq N$,
we have $|\CH(\pi^{(l)}_{R,n}) -\CH(\alpha^{(l)})|<\epsilon$, confirming
\eqref{eq:5_3}.

By Theorem \ref{th:2_5}, we now have
\begin{align*}
\dimfs(R)&=\Dimfs(R)\\
&=\inf_{l\in\Z^+}\frac{1}{l\log k} \limn \CH(\pi^{(l)}_{R,n})\\
&=\inf_{l\in\Z^+}\frac{1}{l\log k} \CH(\alpha^{(l)})\\
&=\frac{\CH(\alpha)}{\log k}\\
&=\CH_k(\alpha).
\end{align*}
\end{proof} 
\end{document}